\documentclass[conference]{IEEEtran}

\usepackage[dvips]{color}
\usepackage{epsf}
\usepackage{epstopdf}
\usepackage{times}
\usepackage{epsfig}
\usepackage{graphicx}
\usepackage{amsmath}
\usepackage{amssymb}
\usepackage{amsxtra}
\usepackage{here}
\usepackage{rawfonts}
\usepackage{times}
\usepackage{amsthm}
\usepackage{url}
\usepackage{cite}

\newtheorem{theorem}{\bf Theorem}

\newtheorem{proposition}{\bf Proposition}

\newtheorem{definition}{\bf Definition}

\newlength{\aligntop}
\newlength{\alignbot}
\makeatletter

\IEEEoverridecommandlockouts
\begin{document}

\title{\huge Many-to-Many Matching Games for Proactive Social-Caching in Wireless Small Cell Networks \vspace{-0.3cm}}

\author{\authorblockN{ Kenza Hamidouche$^1$, Walid Saad$^2$, and  M\'erouane Debbah$^1$} \authorblockA{\small
$^1$ Alcatel-Lucent Chair on Flexible Radio - SUP\'ELEC, Gif-sur-Yvette, France, \\Emails: \url{{kenza.hamidouche,merouane.debbah}@supelec.fr}\\
$^2$ Electrical and Computer Engineering Department, University of Miami, Coral Gables, FL, \\ Email: \url{walid@miami.edu}\vspace{-1cm}
 }%
   \thanks{This research is supported by the National Science Foundation under Grants CNS-1253731 and CNS-1406947.}
 }
\date{}
\maketitle

\vspace{1cm}\begin{abstract}
In this paper, we address the caching problem in small cell networks from a game theoretic point of view. In particular, we formulate the caching problem as a many-to-many matching game between small base stations and service providers' servers. The servers store a set of videos and aim to cache these videos at the small base stations in order to reduce the experienced delay by the end-users. On the other hand, small base stations cache the videos according to their local popularity, so as to reduce the load on the backhaul links. We propose a new matching algorithm for the many-to-many problem and prove that it reaches a pairwise stable outcome. Simulation results show that the number of satisfied requests by the small base stations in the proposed caching algorithm can reach up to three times the satisfaction of a random caching policy. Moreover, the expected download time of all the videos can be reduced significantly. \vspace{-0.2cm}
\end{abstract}

\section{Introduction}\vspace{-0.1cm}
During the last decade, the rapid proliferation of smartphones coupled with the rising popularity of Online Social Networks (OSNs) have led to an exponential growth of mobile video traffic \cite{1}. Meanwhile, existing wireless networks have already reached their capacity limits, especially during peak hours \cite{2}. In order to ensure acceptable Quality of Experience (QoE) for the end-users, the next generation of wireless networks will consist of a very dense deployment of low-cost and low-power small base stations (SBSs) \cite{3}. The SBSs provide a cost-effective way to offload traffic from the main macro-cellular networks. However, the prospective performance gains expected from SBS deployment will be limited by capacity-limited and possibly heterogeneous backhaul links that connect the SBSs to the core network \cite{3}. Indeed, the solution of deploying high speed fiber-optic backhaul is too expensive and thus, the use of DSL backhaul connections  need to be considered \cite{JDSU}. 

Distributed caching at the network edge is considered as a promising solution to deal with the backhaul bottleneck \cite{4,5}. The basic idea is to duplicate and store the data at the SBSs side. Consequently, users' requests can be served locally, from the closest SBSs without using the backhaul links, when possible. 
In prior works, the cache placement problem has been mainly addressed for wired networks, especially for Content Delivery Networks (CDNs) \cite{10, 11, 12, 13}. However, the CDNs topology is different from the small cell networks topology since the CDNs' caching servers are part of the core network while the SBSs are connected to the core network and the User Equipments (UE) via backhaul and radio links, respectively. Hence, it is necessary to explore new cache placement strategies that consider the limited capacity of radio and backhaul links. The placement problem in wireless networks has been studied in \cite{4,5,20} in which a placement scheme that minimizes the expected delay for data recovery has been proposed. The placement strategy is defined without considering the limited capacity of backhaul links. In addition, the assignment of data is based on the global popularity of videos, i.e., the expected proportion of users that would request each video, whereas the popularity of data can differ from an SBS to another SBS depending on the consumption of the users connected to each SBS. Thus, unlike the existing works \cite{4,5,6,7}, which focus on either reducing the download time for the end-users or minimizing the network load, we combine both parameters by considering simultaneously the limited capacity of backhaul and wireless links as well as the induced traffic by each file at the SBSs.

The main contribution of this paper is to develop a novel caching algorithm that aims to reduce the backhaul load and the experienced delay by the end-users when accessing shared videos in OSNs. 
 We leverage the information hosted by OSNs 
to design an accurate \emph{proactive} caching strategy, in which the SBSs \emph{predict} the users' requests and download ahead of time the related videos. We formulate the cache placement problem using \textit{matching theory} which is a game theoretic approach that provides solid mathematical tools suitable to study the considered problem \cite{Roth1990}. While most of the works in wireless networks rely on the \emph{one-to-one} and \emph{many-to-one} matchings \cite{Lesh12,Bres99,Jorwieck2011,14,15,16}, we formulate the caching problem as a \emph{many-to-many matching game} which is less understood in the literature. The formulated game is a matching problem between the two sets of Service Provider Servers (SPSs) and SBSs. Each SPS that hosts a set of videos aims to be matched to a group of SBSs where it can cache related videos. On the other hand, the SBSs are matched to a set of SPSs and indirectly to a set of videos that are stored in those SPSs. Subsequently, the SPSs and SBSs rank one another so as to decide on which files will be cached and at which SBSs.

To solve the \emph{many-to-many game}, we propose a novel, distributed algorithm and prove that it can reach a stable outcome, i.e., no player will have the incentive to leave his partners for other players. 



The rest of this paper is organized as follows: Section~\ref{sec:mod} presents the system model. Section~\ref{sec:form} presents the game formulation and properties. In Section~\ref{sec:sol}, we study the proposed sable matching algorithm. Simulation results are presented and analyzed in Section~\ref{sec:sim}. Finally, conclusions are drawn in Section~\ref{sec:conc}.\vspace{-0.3cm}
\section{System Model and Game Formulation}\label{sec:mod}
Consider two networks, a virtual network and a real network. The virtual network represents an OSN through which $N$ UEs in the set $\mathcal{N}=\{u_1, u_2, u_3,\dots, u_N\}$ are connected to one another via friendship relationships. Thus, these users can interact, communicate and share information with their friends. Suppose that the $N$ UEs share and watch videos chosen from a library of $V$ videos in the set $\mathcal{V}=\{v_1, v_2,\dots,v_V\}$. The service providers that supply the videos store the provided videos in their SPSs. To ensure a better QoE for the end-users, instead of serving users via capacity-limited backhaul links, service providers prefer to store copies of the videos deeper in the network, i.e., at the SBSs, closer to UEs. The real network consists of the $N$ UEs, served by $K$ SPSs in the set $\mathcal{C}=\{c_1, c_2,..., c_K\}$ and $M$ SBSs in the set $\mathcal{M}=\{s_1, s_2, s_3,\dots, s_M\}$. The real network's topology is illustrated in Fig. \ref{model}. Each SPS $c_i$ is connected to an SBS $s_j$ via low-rate backhaul link of capacity $b_{ij}$, using which the SBS downloads videos from that SPS.

%
%
%

The SBSs are equipped with storage units of high but limited storage capacities $Q=[q_1, q_2, q_3, ..., q_M]$, expressed by the number of videos that each SBS can store. Thus, the service providers can cache their videos in the SBSs such that each SBS $s_i$ can locally serve a UE $u_j$ via a radio link of capacity $r_{ij}$.

\begin{figure}
	\center
	\includegraphics[scale=0.25]{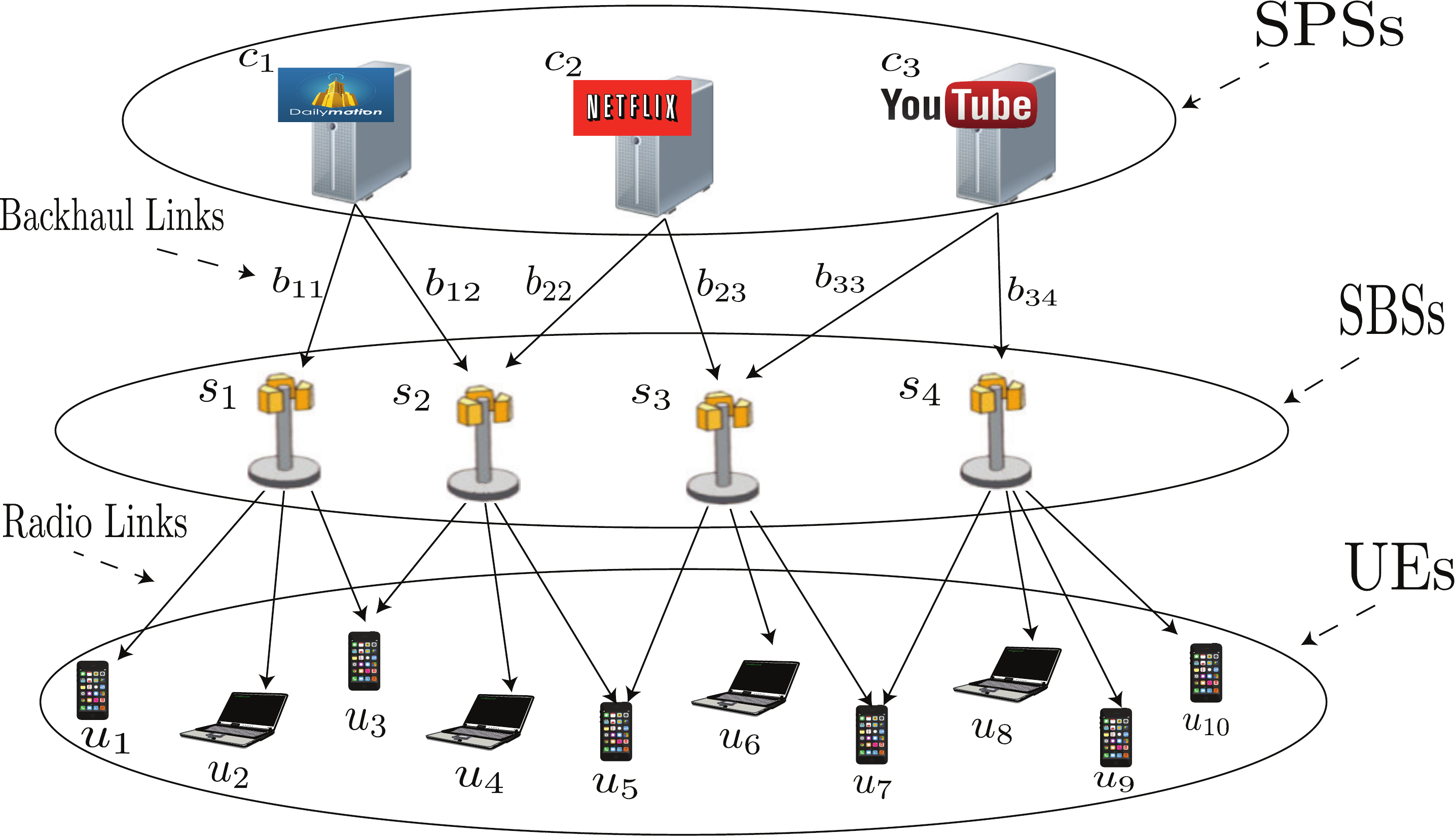}
	 	\caption{Illustration of the proposed network model.}
	\label{model}
\end{figure}

In this scenario, our goal is to produce a proactive download of video content at the SBSs level. A caching is said proactive if the SBSs can predict the users' requests and download ahead of time the related videos. Each SBS captures users' requests for the shared videos based on users' interests and interactions in the OSN. The most important properties of social content that can be used to design efficient proactive caching strategies are discussed next.

\subsection {\textit{ Social Interactions}} 
The videos that a user typically watches depend strongly on the friends who share them. In fact, a user is more likely to request a video, shared by one of his friends, if this user is used to watch the shared videos by that friend \cite{soc}. This induced popularity can be given by: 
$I_{\textrm{social}}=\frac{\alpha_{ln}}{\sum_{j=1}^{F_l}{\alpha_{jl}}}$,
where $\alpha_{ln}$ is the number of videos previously shared by user $u_n$ and viewed by user $u_l$. $F_l$ is the number of $u_l$'s friends.

\subsection{\textit{ Sharing Impact}}
Given that each video file can belong to a distinct category (e.g., news, music, games, etc.), we let $S_{gl}$ the number of videos of category $g$ shared by a user $u_l$. Whenever a user $u_l$'s request for a specific video is predicted and cached in its serving SBS $s_m$, sharing this video with this user's friends can have an important impact on the traffic load. This sharing impact depends on the number of user $u_l$'s friends that are connected to the same SBS and the probability that user $u_l$ shares the video. More formally, the sharing impact is given by:
$I_{\textrm{sharing}}= F_l^m\cdot  \frac{S_{gl}}{\sum_{i=1}^{G}{S_{il}}}$,
where $F_l^m$ is the number of $u_l$'s friends that are connected to the SBS $s_m$ and $G$ is the total number of the considered video categories.
\subsection{\textit{Users' Interests }}
Whenever a user is interested in a certain topic, it can request a video that belongs to its preferred categories irrespective of the friend who shared it \cite{soc},\cite{soc2}. Based on the categories of the previously watched videos by a user $u_l$, an SBS can predict the user $u_l$'s interests. The impact of this parameter is computed using:
$I_{\textrm{Interests}}= \frac{V_{gn}}{\sum_{i=1}^{H}{V_{il}}}$,
with $V_{gn}$ being the number of viewed videos of category $g$ by a user $u_n$, and $H$ being the number of videos in the history of user $u_l$.

Given these factors, our goal is to predict users' requests and accordingly select and cache a set of videos ateach SBS. Hence, we formulate this caching problem as a \emph{many-to-many matching} game, in which the SPSs aim to cache their videos in the SBSs that offer the smallest download time for the requesting users, while the SBS prefer to cache the videos that can reduce the backhaul load.   
\section{Proactive-Caching as a Many-to-Many Matching Game}\label{sec:form} 
\subsection{Matching Concepts}
To model the system as a many-to-many matching game \cite{Soto1999}, we consider the two sets $\mathcal{C}$ of SPSs and $\mathcal{M}$ of SBSs as two teams of players. The \textit{matching} is defined as an assignment of SPSs in $\mathcal{C}$ to SBSs in $\mathcal{M}$. The SPSs acts on behalf of the video files and each of them decides on its own videos. Meanwhile, SBSs store videos depending on their storage capacity. In a matching game, the storage capacity of an SBS $s$ as well as the number of SBSs, in which an SPS $c$ would like to cache a given file $v$ are known as the quotas, $q_s$ and $q_{(c,v)}$, respectively \cite{Roth1990}. Since an SPS decides in which it caches a video $v$ independently from the other owned files, for ease of notation, we use $v$ instead of $(c,v)$ pairs. 
\begin{definition}
A \emph{many-to-many matching} $\mu$ is a mapping from the set $\mathcal{M}\cup \mathcal{V}$ into the set of all subsets of $\mathcal{M}\cup \mathcal{V}$ such that for every $v \in \mathcal{V}$ and $s \in \mathcal{M}$ \cite{Roth1991}:
\begin{enumerate}
\item $\mu (v)$ is contained in $\mathcal{S}$ and $\mu(s)$ is contained in $\mathcal{V}$; 
\item $|\mu (v)| \leq q_v $ for all $v$ in $\mathcal{V}$;
\item $|\mu (s)| \leq q_s$ for all $s$ in $\mathcal{S}$;
\item $s$ is in $\mu (v)$ if and only if $v$ is in $\mu(s)$,
\end{enumerate}
\end{definition}
with $\mu (v)$ being the set of player $v$'s partners under the matching $\mu$.
 
The definition states that a matching is a many-to-many relation in the sense that each stored video in an SPS is matched to a set of SBSs, and vice-versa. In other words, an SPS can decide to cache a video in a number of SBSs and an SBS cache videos originating from different SPSs. 
Before setting an assignment of videos to SBSs, each player needs to specify its \emph{preferences} over subsets of the opposite set based on its goal in the network. We use the notation $S \succ_m T$ to imply that SBS $m$ prefers to store the videos in the set $S \subseteq \mathcal{V}$ than to store the ones proposed in $T \subseteq \mathcal{V}$. A similar notation is used for the SPSs to set a preference list for each video. Faced with a set $S$ of possible partners, a player $k$ can determine which subset of $S$ it wishes to match to. We denote this choice set by $C_k(S)$. Let $A(i,\mu)$ be the set of $j \in \mathcal{V} \cup \mathcal{M}$ such that $i \in \mu(j)$ and $j \in \mu (i)$.

To solve the matching game, we are interested to look at a stable solution, in which there are no players that are not matched to one another but they all prefer to be partners. In many-to-many models, many stability concepts can be considered depending on the number of players that can improve their utility by forming new partners among one another. However, the large number of SBSs which is expected to exceed the number of UEs \cite{Number}, makes it difficult to identify and organize large coalitions than to consider pairs of players and individuals. Hence, in this work we are interested in the notion of pairwise stability, defined as follows \cite{Roth1984}: 

\begin{definition}
A matching $\mu$ is \emph{pairwise stable} if there does not exist a pair $(v_i,s_j)$ with $v_i\notin\mu(s_j)$ and $s_j\notin \mu(v_i)$ such that $T\in C_{v_i}(A(v_i,\mu) \cup \{s_j\})$ and $S\in C_{s_j}(A(s_j,\mu) \cup \{v_i\})$ then $T\succ_{v_i} A(v_i,\mu)$ and $S\succ_{s_j} A(s_j,\mu)$
\end{definition}

In the studied system, the SBSs and SPSs are always interested in the gain they can get from individuals of the opposite set. For instance, an SBS would always like to first cache the most popular file as long as that file is proposed to it. Thus, even though the set of stable outcomes may be empty \cite{Soto1999}, SPSs and SBSs have \emph{substitutable} preferences, defined as follows \cite{kelso1982}:

\begin{definition}
Let $T$ be the set of player $i$'s potential partners and $S \subseteq T$. Player $i$'s preferences are called \emph{substitutable} if for any players $k$, $k^{\prime} \in C_i(S)$ then $k\in C_{i}(S\setminus \{k^{\prime}\})$.
\end{definition}
In fact, \emph{substitutability} is the weakest needed condition for the existence of a pairwise stable matching in a many-to-many matching game \cite{Roth1984}. 

\subsection{Preferences of the Small Base Stations\vspace{-0.2ex}}
Based on the social features previously discussed, we define the local popularity of a video $v_i$ at the $m^{th}$ SBS as follows:   
\begin{equation}
P_{v_i}= \sum_{l=1}^{F_n^m}{I_{\textrm{sharing}} (\gamma \cdot I_{\textrm{social}} + (1-\gamma) I_{\textrm{Interests}})},
\end{equation}
where $\gamma\in [0,1]$ is a weight that balances the impacts of social interactions and users' interests on the local popularity of a video.
\subsection{Preferences of the Service Provider Servers}
The goal of service providers is to enhance the quality experienced by the end-users. In fact, an SPS $c_i$ would prefer to cache a video $v_k$ at the SBS $s_j$ that offers the smallest download time for the expected requesting UEs. The download time depends on the capacity of the backhaul link $b_{ij}$ and the radio links $r_{jn}$ that connect the SBS $s_j$ to the UE $u_n$. The video file is first downloaded by $s_j$ which then serves the UEs. Thus, in the worst case, downloading a video stored in $c_i$ takes the required time to pass by the link with the poorest capacity. When many UEs are expected to request the same file from $s_j$, the download time is given by:
\begin{equation}
T_D = \frac{1}{\textrm{min}(b_{ij}, \frac{\sum_{n=1}^{N}r_{jn}}{N})}.
\end{equation}
Since each video might be requested by different UEs, an SPS defines its preferences over the SBSs for each owned video file.

\section{Proactive Caching Algorithm}\label{sec:sol}\vspace{-0.2cm}
To our knowledge, \cite{Infocom} is the only work that deals with many-to-many matchings in wireless networks. The proposed algorithm deals with \emph{responsive} preferences which is a stronger condition than \emph{substitutability}. Thus, the algorithm could not be applied to our case. Under substitutable preferences, a stable matching algorithm has been proposed in \cite{Roth1990} for many-to-one games. The pairwise stable matching in the many-to-many problem, has been proved to exist between firms and workers when salaries (money) are explicitly incorporated in the model \cite{Roth1984}. Here, we extend and adapt these works to our model in order to propose a new matching algorithm that is proven to reach a pairwise stable.

\subsection{Proposed Algorithm}
After formulating the caching problem as a many-to-many game, we propose an extension of the \emph{deferred acceptance} algorithm \cite{GaleandShapley} to the current model with SPSs proposing. The algorithm consists of three phases. During the first phase, SPSs and SBSs discover their neighbors and collect the required parameters to define the preferences, such as the backhaul and radio links capacities. This can be done for instance, by exchanging \emph{hello} messages periodically. In the second phase, SPSs define a preference list for each owned file over the set of SBSs, while the SBSs define their preferences over the set of videos that would be proposed by the SPSs. The last phase consists of two steps. In the first step, every SPS proposes an owned video to the most preferred set of SBSs that offer the shortest download time for that video. Afterwards, each SBS $s_j$ rejects all but the $q_j$ most popular videos from the set of alternatives proposed to it. In the second step, the SPSs propose an owned video to the most preferred set of SBSs, which includes the SBSs to which it previously proposed that video and have not yet rejected it (\emph{substitutability}). Each SBS rejects all but its choice set from the proposed videos. The second step is repeated until no rejections are issued. The algorithm is summarized in Table \ref{tab:algo1}.

\begin{table}[!t]
\scriptsize
  \centering
  \caption{
    \vspace*{-0.3em}Proposed Proactive Caching Algorithm}
    \begin{tabular}{p{8cm}}
      \hline

\vspace{-0.1cm}\textbf{Phase 1 - Network Discovery:}   \vspace*{.1em}\\
\hspace*{1em}-Each SPS discovers its neighboring SBSs and collects the required network \vspace*{.1em}\\
\hspace*{1em}parameters.\vspace*{.1em}\\
\textbf{Phase 2 - Specification of the preferences}   \vspace*{.1em}\\
\hspace*{1em}-Each SPS and SBS sets its preference list(s).\vspace*{.1em}\\
\textbf{Phase 3 - Matching algorithm}   \vspace*{.1em}\\
\hspace*{1em}-The SPSs propose each owned video to the SBSs in $C_{v_i}(\mathcal{M})$ to cache it.\vspace*{.1em}\\
\hspace*{1em}-Each SBS rejects all but the $q_s$ most preferred videos.\vspace*{.1em}\\
\hspace*{1em}\textbf{Repeat}\vspace*{.1em}\\
\hspace*{2em}-The SPSs propose each file to the related most preferred set of SBSs \vspace*{.3em}\\
\hspace*{2em}that includes all of those SBSs whom it previously proposed to and \vspace*{.3em}\\
\hspace*{2em} who have not yet rejected it (\textbf{\emph{Substitutability}}). \vspace*{.3em}\\
\hspace*{2em}-Each SBS rejects all but the $q_s$ most preferred videos.\vspace*{.1em}\\
\hspace*{1em}\textbf{Until} convergence to a stable matching\vspace*{.1em}\\
   \hline
    \end{tabular}\label{tab:algo1}\vspace{-0.2cm}
\end{table}

\subsection {Pairwise Stability}
Let $\mathcal{P}_{s_j}(k)$ be the set of proposals received by an SBS $s_j$ at step $k$, and $C_{v_i}(\mathcal{M}, k)\subseteq \mathcal{M}$ be the choice set of SBSs to which a video $v_i$ has been proposed at step $k$. By analyzing the algorithm, we state the following propositions.

\begin{proposition}
\emph{Offers remain open}: For every video $v_i$, if an SBS $s_j$ is contained in $C_{v_i}(\mathcal{M},(k-1))$ at step $k-1$ and did not reject $v_i$ at this step, then $s_j$ is contained in $C_{v_i}(\mathcal{M},k)$.
\end{proposition}

\begin{proof}
Note that $C_{v_i}(\mathcal{M}, k-1)=C_{v_i}(C_{v_i}(\mathcal{M}, k-1)\cup C_{v_i}(\mathcal{M}, k))$, since $C_{v_i}(\mathcal{M}, k-1)$ is video $v_i$'s choice set from all those sets whose elements have not been rejected prior to step $k-1$, while $C_{v_i}(\mathcal{M}, k)$ is the choice set from the smaller class of sets whose elements have not been rejected prior to step $k$. Here, substitutability implies that if $s_j\in C_{v_i}(\mathcal{M}, k-1)$, then $s_j\in C_{v_j}(C_{v_i}(\mathcal{M}, k)\cup \{s_j\})$. So if $s_j \in C_{v_i}(\mathcal{M}, k-1)$ is not rejected at step $k-1$, it must be contained in $C_{v_i}(\mathcal{M}, k)$, since otherwise $C_{v_i}(C_{v_i}(\mathcal{M}, k)\cup \{s_j\}) \succ_{v_i} C_{v_i}(\mathcal{M}, k)$, violating the requirement that $C_{v_i}(\mathcal{M}, k)$ is the most preferred set whose elements have not been rejected.
\end{proof}

\begin{proposition}
\emph{Rejections are final}: If a video $v_i$ is rejected by an SBS $s_j$ at step $k$ then at any step $p \geq k$, $v_i \notin C_{s_j}(\mathcal{P}_{s_j}(p) \cup \{v_i\})$.
\end{proposition}
\begin{proof}
Assume that the proposition is false, and let $p \geq k$ be the first step at which $v_i \in C_{s_j}(\mathcal{P}_{s_j}(p) \cup \{v_i\})$. Since $C_{s_j}(\mathcal{P}_{s_j}(p))$ contains $C_{s_j}(\mathcal{P}_{s_j}(p-1))$ by Proposition 1, substitutability implies that $v_i\in C_{s_j}(C_{s_j}(\mathcal{P}_{s_j}(p-1))\cup \{v_i\})$ and thus $v_i \in C_{s_j}(\mathcal{P}_{s_j}(p-1)\cup \{v_i\})$, which contradicts the definition of $p$ and completes the proof of the proposition.  
\end{proof}

\begin{theorem}
The proposed matching algorithm between SPSs and SBSs is guaranteed to converge to a pairwise stable matching.
\end{theorem}
\begin{proof}
We make the proof by contradiction. Suppose that there exists a video $v_i$ and an SBS $s_j$ with $v_i\notin\mu(s_j)$ and $s_j\notin \mu(v_i)$ such that $T\in C_{v_i}(A(v_i,\mu) \cup \{s_j\})$, $S\in C_{s_j}(A(s_j,\mu) \cup \{v_i\})$ and $T\succ_{v_i} A(v_i,\mu)$ and $S\succ_{s_j} A(s_j,\mu)$. Then, the video $v_i$ was proposed to the SBS $s_j$ which rejected it at some step $k$. So $v_i\notin C_{s_j}(A(s_j,\mu)\cup \{v_i\})$, and thus $\mu$ cannot be unstable. 
\end{proof}
\section{Simulation Results and Analysis}\label{sec:sim}\vspace{-0.1cm}
For simulations, we consider a network in which we set the number of SPSs, SBSs, UEs and videos to $K=80$, $M=150$, $N=400$ and $V=100$, respectively. We assume that all the SBSs have the same storage capacity, while the backhaul links capacities are lower than radio links capacities which captures the real network characteristics \cite{JDSU}. We set the total backhaul links capacities and radio links capacities to $B=80 \textrm{ Mbit/time slot}$ and $R=180 \textrm{ Mbit/time slot}$, respectively. Due to the unavailability of social dataset that includes the addressed parameters in this paper (e.g., the accessed video categories and shared videos by users) because of privacy reasons, we generate files popularity and user requests pseudo-randomly. The popularity of files is generated at each SBS with a Zipf distribution which is commonly used to model content popularity in networks \cite{Bres99}. Users' requests are generated using a uniform distribution.

To evaluate the performance of the proactive catching algorithm, we implement the proposed matching algorithm (MA) as well as a random caching policy (RA), in which the SBSs are filled randomly with videos from the SPSs to which they are connected to, until reaching their storage capacity limit. We run the two algorithms for different values of a storage ratio $\beta$, which represents the number of files that each SBS has the capacity to store. More formally, $\beta =\frac{q_i}{V} \forall i\in V$. In the numerical results, we show the evolution of the satisfaction ratio, which corresponds to total number of served requests by the SBSs over the total number of requests, while increasing the number of requests in the network. Moreover, we compare the mean required time for downloading all the videos by the requesting users.      

\begin{figure}[!t]
  \begin{center}
    \includegraphics[scale=0.48]{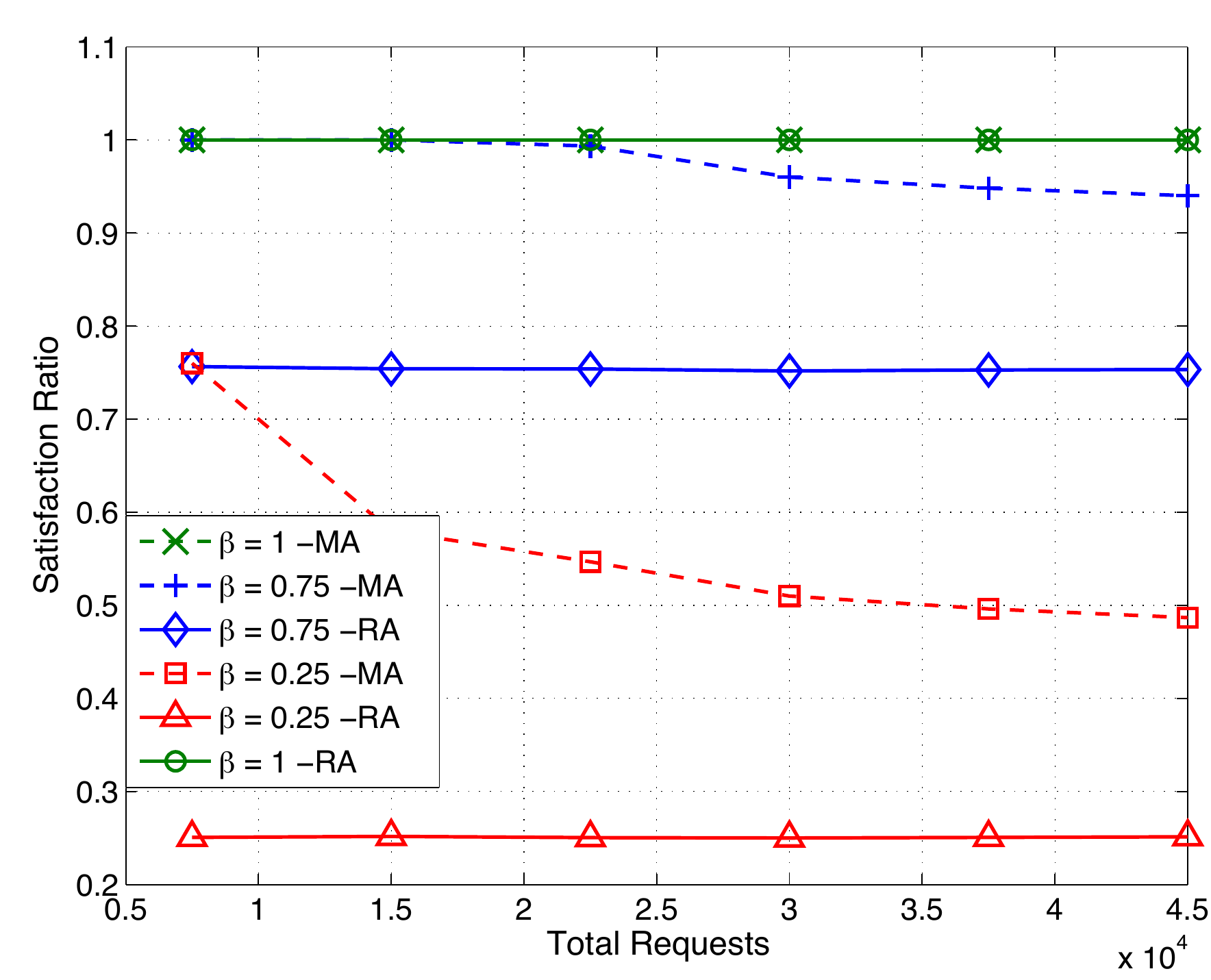}
    \vspace{-0.3cm}
    \caption{\label{rat} Satisfaction evolution for MA and RA.}
        \vspace{-0.6cm}
  \end{center}
\end{figure}

In Fig. \ref{rat}, we show the proportion of served users for $\beta \in\{0.25, 0.75, 1\}$. The satisfaction ratio decreases when the storage capacity of the SBSs decreases. This is evident as less files are cached in the SBSs. Fig. \ref{rat} shows that, when each SBS has the capacity to store all the proposed videos by the SPSs ($\beta=1$), the \emph{satisfaction ratio} remains equal to 1 irrespective of the used caching policy. This result stems from the fact that all the requests are served locally by the SBSs. When the SBSs have the capacity to store $25\%$ and $75\%$ of the proposed videos by the SPSs, i.e., $\beta=0.25$ and $\beta =0.75$ respectively, the satisfaction of the MA is up to three times higher than the RA ($\beta=0.25$). Under the MA, the number of served users by the SBSs decreases by increasing the number of requests. This is due to the fact that SBSs choose to cache first the files with a higher local popularity. Fig. \ref{rat} shows that, as the number of requests increases, the satisfaction ratio of the RA changes only slightly due to the uniform selection of videos under the same distribution of files popularity. 

Fig. \ref{tim} shows the expected download time of all the files in the network for $\beta \in \{0.25, 1\}$. The download time is lower when the requests are served by the SBSs ($\beta =1$) compared to the case in which some requests need to be served by the SPSs ($\beta=0.25$). This is due to the fact that those videos would require more time to pass through the backhaul links which are of lower capacities compared to the radio links. In Fig. \ref{tim}, we can see that, when $\beta=1$, the expected download time is identical for both RA and MA. When $\beta=0.25$, although the download time increases by increasing the number of requests in the network due to the network congestion, the MA outperforms the RA. In fact, the higher satisfaction ratio of the MA compared to RA leads to a smaller expected download time because most of the requests are served via the radio links.      
\begin{figure}[!h]
	\center
	\includegraphics[scale=0.48]{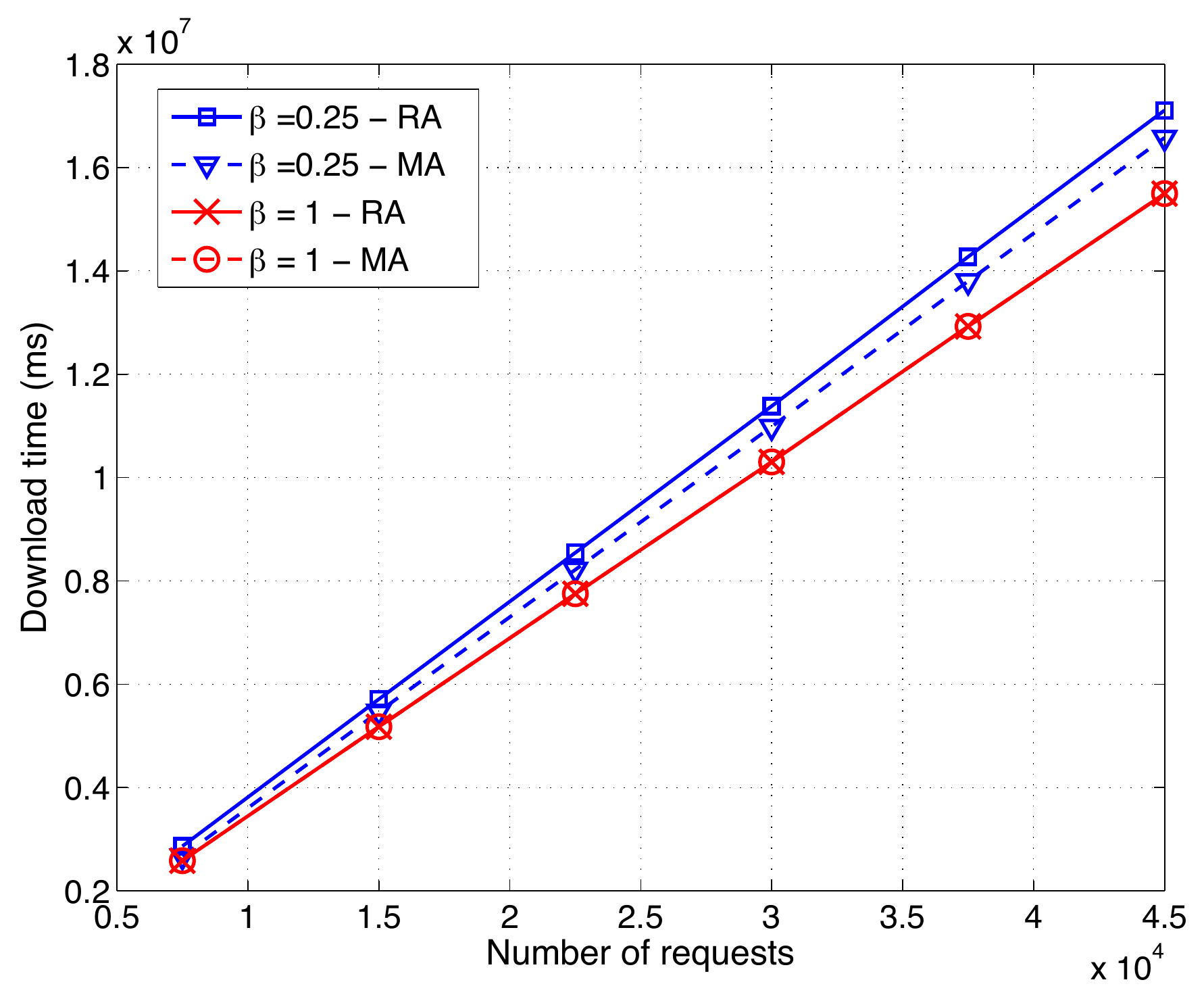}
	  \vspace{-0.3cm}
	 	\caption{Download time as the number of requests increases.}
	 	\vspace{-0.3cm}
	\label{tim}
\end{figure}

As seen from Figs. \ref{rat} and \ref{tim}, clearly, the proposed MA allows to efficiently overcome the backhaul bottleneck and improve the performance of video download in small cell networks.

\section{Conclusion}\label{sec:conc} \vspace{-0.2cm}
In this paper, we have proposed a novel caching approach for overcoming the backhaul capcity constraints in wireless small cell networks. We have formulated a many-to-many matching game while considering the limited capacity of backaul links as well as the impact of the local popularity of each video at the SBSs. To solve the game, we have proposed a new matching algorithm that assigns a set of videos to each SBSs. Simulation results have shown that the proposed matching game enables the SBSs and SPSs to decide strategically on a cache placement that reduces the backhaul links load as well as the experienced delay by the end-users.\vspace{-0.2cm}

\bibliographystyle{IEEEtran}
\bibliography{references}

\end{document}